\documentclass[accepted]{uai2026} 
                        
\usepackage{graphicx, subfig}
\usepackage[american]{babel}
\usepackage[colorinlistoftodos]{todonotes}
\usepackage{natbib} 
\usepackage{mathtools} 
\usepackage{booktabs} 
\usepackage{tikz} 
\usepackage{amsfonts}
\usepackage{amssymb}
\usepackage{amsthm}
\usepackage{todonotes}
\usepackage{algorithm}     
\usepackage{algpseudocode}   
\usepackage{hyperref}
\newtheorem{proposition}{Proposition}
\newtheorem{definition}{Definition}[section]
\definecolor{men}{RGB}{42, 82, 152}      
\definecolor{women}{RGB}{181, 54, 58}    
\definecolor{cross}{RGB}{90, 122, 74}    
\usepackage[most]{tcolorbox}
\usepackage{adjustbox}

\newlength{\subcolumnwidth}
\newenvironment{subcolumns}[1][0.45\columnwidth]
 {\valign\bgroup\hsize=#1\setlength{\subcolumnwidth}{\hsize}\vfil##\vfil\cr}
 {\crcr\egroup}
\newcommand{\nextsubcolumn}[1][]{%
  \cr\noalign{\hfill}
  \if\relax\detokenize{#1}\relax\else\hsize=#1\setlength{\subcolumnwidth}{\hsize}\fi
}

\title{How Predicted Links Influence Network Evolution:\\Disentangling Choice and Algorithmic Feedback in Dynamic Graphs}

\author[1]{\href{mailto:<mathilde.perez@telecom-paris.fr>?Subject=Your UAI 2026 paper}{Mathilde Perez}{}}
\author[2]{Raphaël Romero}
\author[2]{Jefrey Lijffijt}
\author[1]{Charlotte Laclau}
\affil[1]{%
    LTCI\\
    Télécom Paris\\
    Institut Polytechnique de Paris\\
    France.\\
}
\affil[2]{%
    IDLab \\
    Ghent University \\
    Belgium.
}

\begin{document}
\maketitle

\begin{abstract}
Link prediction models are increasingly used to recommend interactions in evolving networks, yet their impact on network structure is typically assessed from static snapshots. In particular, observed homophily conflates intrinsic interaction tendencies with amplification effects induced by network dynamics and algorithmic feedback. We propose a temporal framework based on multivariate Hawkes processes that disentangles these two sources and introduce an instantaneous bias measure derived from interaction intensities, capturing current reinforcement dynamics beyond cumulative metrics. We provide a theoretical characterization of the stability and convergence of the induced dynamics, and experiments show that the proposed measure reliably reflects algorithmic feedback effects across different link prediction strategies.
\end{abstract}

\section{Introduction}\label{sec:intro}
Social networks are known to be homophilic, meaning that \textit{similar} individuals are more likely to be connected \cite{mcphersonBirdsFeatherHomophily2001}. 
Following \cite{Kossinets2009}, this homophily can be decomposed into two distinct components, namely \emph{choice homophily} and \emph{induced homophily}. Choice homophily refers to the tendency of individuals to form or maintain ties as a function of their own characteristics, such as demographic attributes, interests, values, or social status, independently of the existing network structure. This tendency does not necessarily imply a conscious or deliberate preference for similarity; it may arise from latent preferences, social dispositions, or constraints associated with individuals themselves. For this reason, choice homophily is typically viewed as intrinsic to individuals. 
Induced homophily, on the other hand, arises from the structure and dynamics of the system in which interactions take place. It is driven by factors such as network topology, time and order of interactions, or algorithmic mechanisms that shape visibility and opportunities for interaction. In practice, network structure and algorithmic intervention tend to reinforce homophily, even when individual preferences remain stable. Local structural effects \citep{asikainen2020cumulative, granovetter1973strength, currarini2009economic} and repeated exposure mechanisms \citep{chaney2018algorithmic, pariser2011filter, flaxman2016filter} can amplify existing patterns of connections, leading to increasingly segregated interaction patterns. 

Beyond its sociological interpretation, this distinction between choice and induced homophily has important implications for how learning algorithms on graphs are evaluated and designed, yet it is rarely made explicit in current graph learning pipelines. In link prediction (LP) and node classification, models are trained to exploit the connectivity patterns present in the observed graph. When interaction patterns are segregated, learning algorithms tend to reproduce and sometimes amplify these patterns. Homophily is therefore often used as a proxy for bias in graph-based learning. This perspective is also reflected in the literature, where many debiasing approaches, whether based on graph rewiring \citep{laclau2020optimal, spinelli2021}, the use of explicit priors \citep{buyl2020}, or on Graph Neural Networks (GNNs) \citep[for instance]{zheng2024, zhang2025}, explicitly aim to reduce homophily with respect to a given \textit{sensitive} attribute. 

In this paper, we study how separating intrinsic and induced sources of homophily matters for the design and evaluation of (fair) link prediction models. We believe that this distinction is important to reason about observed disparities and potential sources of unfairness in graph-based learning. 
To proceed, we move beyond a static view of the graph and consider its temporal evolution. Indeed, real-world networks are dynamic, and link prediction models directly affect the network's future structure by shaping which connections are formed. The need to adopt this dynamic perspective also stems from several recent studies that explore the long-term impact of fairness in dynamic settings, using simulated environments \citep{damour2020fairness, 10.5555/3666122.3668529, rateike2024designing} or the repeated empirical risk minimization framework \citep{liu2018delayed, hashimoto2018fairness}. These works emphasize that accounting for feedback mechanisms is crucial and that fairness interventions can have delayed and counterintuitive effects over time. 
Next, we detail the most related contributions, which are at the intersection of fairness and dynamic graphs. 

\paragraph{Related work}
\cite{asikainen2020cumulative} and \cite{Vega-Oliveros2019} are the studies closest to our setting, as both analyze the interaction between algorithmic mechanisms and network structure in dynamic graphs. \cite{asikainen2020cumulative} show that weak intrinsic preferences can be amplified by endogenous processes such as triadic closure, leading to cumulative segregation. We share the distinction between intrinsic and induced homophily and a dynamic viewpoint, but depart from their framework by modeling interactions with Hawkes processes and explicitly integrating link prediction (LP) into the network evolution, enabling a fairness-oriented analysis of algorithmic mediation. \cite{Vega-Oliveros2019} study how different link recommendation strategies affect information diffusion and induce topological changes in the network, under various diffusion models. Their analysis is primarily empirical and does not explicitly model homophily or group-level disparities.

From a complementary angle, closely related work in the field of recommender systems \citep{akpinar2022longterm, fabbri2022exposure, cao2024recommendation} studies fairness in terms of exposure and visibility over time, highlighting the joint role of homophily, group size imbalance, and algorithmic feedback. While these approaches emphasize the importance of temporal effects, they typically abstract away the dynamics of interactions within the network itself, or rely on algorithm-dependent notions of fairness. 

Finally, several works leverage Hawkes processes to model temporal interactions in dynamic graphs and to learn time-aware node representations \citep{pmlr-v129-liu20a, wang2025-HawkesGNN}. Although these methods focus on representation learning rather than on homophily or fairness, they further support the relevance of Hawkes-type models for capturing the dynamics of evolving networks.

\paragraph{Contributions.}
Our work is guided by the following central question: \emph{how do intrinsic interaction and algorithmic feedback jointly shape the long-term evolution of homophily and structural disparities in dynamic networks?}\\
To address this question, we make the following contributions.
(i) We propose a Hawkes-process-based model of network evolution that explicitly separates baseline interaction propensities from reinforcement effects driven by past interactions and exposure, enabling a principled distinction between choice and induced homophily.
(ii) We provide a tractable mean-field characterization of the temporal dynamics of homophily and group-level disparities, allowing us to analyze stability and long-term behaviors under different reinforcement regimes and algorithmic feedback mechanisms.
(iii) We study the short- and long-term impact of fairness-aware LP strategies on network structure, highlighting how feedback loops can induce delayed effects on homophily and structural disparities.
These contributions are reflected in the structure of the paper, from modeling to theory and empirical evaluation.


\section{Motivation and Background} 
For clarity, we first review the general setting of (fair) LP and its connection to homophily in graphs. We then briefly introduce the general framework of Hawkes processes on graphs, which forms the basis of our model.

\subsection{Fairness on Graphs and Homophily}\label{sec:back-fairgraph}

\paragraph{Link prediction and dyadic fairness.}
We consider a graph $\mathcal{G} = (\mathcal{V}, \mathcal{E})$ where each node $u \in \mathcal{V}$ has a sensitive attribute $S_u \in \mathbb{S}$ indicating its demographic group (e.g., $\mathbb{S} = \{0, \ldots, K\}$, where $K$ is the number of groups). The LP task consists in learning a function $h: \mathcal{V} \times \mathcal{V} \to [0,1]$ that estimates the probability of an edge between node pairs.


Group fairness at the interaction level, commonly referred to as \emph{dyadic fairness} \citep{li2021dyadic}, aims to ensure balanced treatment of within-group and cross-group interactions. A standard metric is the demographic parity gap:
\begin{align*}
\Delta_{\text{DP}} = \big|\mathbb{P}&(h(u,v) = 1 \mid S_u = S_v) \\
&- \mathbb{P}(h(u,v) = 1 \mid S_u \neq S_v)\big|,
\end{align*}
where $S_u$ and $S_v$ denote the groups of nodes $u$ and $v$. This metric directly quantifies disparities in predicted link rates between same-group and cross-group pairs.

\paragraph{Homophily as a proxy for bias.}
Dyadic fairness metrics are closely related to homophily with respect to the sensitive attribute $S$. While there is no consensus on how to measure homophily \citep{berryEstimatingHomophilySocial2021}, various structural measures have been proposed. Assortativity \citep{Newman_2002} captures the tendency of nodes to connect to others in the same group, for example.
Since LP models are trained to reproduce observed connectivity patterns, homophilic structures naturally lead algorithms to favor within-group predictions. For this reason, homophily is the main driver of biased LP \citep{marey2026topofair}, and many bias mitigation strategies explicitly target homophily reduction \citep{buyl2020, spinelli2021, zheng2024, zhang2025}.


\paragraph{Limitations of static measures.}
Both homophily measures and fairness metrics like demographic parity are static snapshots that do not distinguish choice homophily (intrinsic preferences) from induced homophily (amplification through network feedback and algorithmic recommendations \citep{asikainen2020cumulative, chaney2018algorithmic}). But we argue that this distinction matters: an algorithm may satisfy a certain level of demographic parity initially yet amplify bias as the network evolves, or a fairness intervention may require extended observation before showing measurable effects. This motivates us to model the temporal evolution of homophily to separate intrinsic tendencies from feedback-driven dynamics.

\subsection{Hawkes Processes on Graphs}
Hawkes processes are a class of self-exciting temporal point processes that can be used to model interaction dynamics in evolving networks, where past events increase the likelihood of future events \citep{yang2017decoupling,huang2022mutually,Perez2025}. They are particularly well-suited to settings where interactions exhibit temporal dependence and feedback effects.

We model network dynamics as a sequence of edge activation events. Each event occurs at time $t \ge 0$ and is associated with a \emph{mark} $m$ from a mark space $\mathcal{M}$, encoding the type of interaction (e.g., the pair of nodes involved). The event sequence is represented as a marked point process
\begin{equation}\label{eq:N}
    N = \sum_{\text{events}} \delta_{(t, m)}.
\end{equation}
For any interval $[s,t)$ and subset of marks $A \subseteq \mathcal{M}$, $N([s,t) \times A)$ counts the number of events of type $A$ occurring between times $s$ and $t$. We write $N(ds, m')$ for the infinitesimal increment associated with events of mark $m'$ at time $s$.


The dynamics of the process are governed by its conditional intensity function $\lambda(t,m)$, which specifies the instantaneous rate at which an event with mark $m$ occurs at time $t$, given the history of past events. For Hawkes processes, the intensity takes the form
\[
\lambda(t,m) = \mu(m) + \sum_{m' \in \mathcal{M}} \int_{[0,t)} \phi\big(m',\, t - s,\, m\big)\, N(ds, m'),
\]
where $\mu(m)$ is the base rate that captures the exogenous propensity for events and $\phi$ is an excitation kernel that models how past events influence the future rate of events.

This formulation can capture temporal clustering, endogenous feedback, and structured interaction patterns, making Hawkes processes well-suited for studying how homophily evolves under algorithmic interventions.

\section{Modeling Framework}
We present our framework for modeling the temporal dynamics of homophily while explicitly distinguishing between choice homophily and induced homophily. 

\subsection{Set Up and Notation} \label{sec:cadre-notation}

We consider a temporal graph with a fixed node set $\mathcal{V}$. To analyze fairness-related structural effects, we introduce a group structure on the nodes: each node $u \in \mathcal{V}$ belongs to a latent group $g_u \in \{1, \ldots, K\}$. These groups may represent demographic attributes (e.g., gender, ethnicity), communities or other structural properties, and may be observed or latent. In the context of fairness, groups typically correspond to the sensitive attribute introduced in Section \ref{sec:back-fairgraph}.

We also assume that the graph is undirected. Edges appear over time as instantaneous events: each edge activation is characterized by a pair of nodes $\{u,v\}$ and a timestamp. We use the notation $\{u,v\}$ to denote the unordered pair of nodes, characterized by their group memberships $\{g_u, g_v\}$. For any pair of groups $i, j \in \{1,\ldots,K\}$ with $i \leq j$, we denote by $\mathcal{E}_{ij}$ the set of edges connecting groups $i$ and $j$:
\[
\mathcal{E}_{ij} = \left\{ \{u,v\} : g_u = i, g_v = j \right\}.
\]


The collection $\{\mathcal{E}_{ij}\}_{1 \le i \le j \le K}$ partitions the edge set according to group interactions. Edges with $i = j$ are \emph{within-group edges}; edges with $i \neq j$ are \emph{cross-group edges}. Without loss of generality, this formulation enables fine-grained analysis of structural patterns and fairness disparities across all group pairs.


\subsection{Group-Structured Hawkes Process} \label{sec:hawkes_homophily}

To account for group structure, we extend the Hawkes process framework introduced previously. We model network evolution as a group-structured marked point process, where marks correspond to group pairs and each event represents an edge activation. This formulation assumes that the interaction dynamic is homogeneous between within-group node pairs and cross-group pairs.


\paragraph{Group-based marks.}
Since the graph is undirected, each edge activation event involves an unordered pair of nodes $\{u,v\}$, characterized by the group pair $(g_u,g_v)$ with $g_u\leq g_v$. We denote by $\mathcal{P} = \{(i,j): 1 \le i \le j \le K\}$ the set of all possible group pairs. Group pairs $(i,j) \in \mathcal{P}$ serve as marks, so the mark space is $\mathcal{M} = \mathcal{P}$.
For each group pair $(i,j)$, we define the associated counting process
\[
N_{ij}(t) = N\big([0,t) \times \{(i,j)\}\big),
\]
which counts the number of interaction events between nodes from groups $i$ and $j$ up to time $t$.

\paragraph{Group-pair Hawkes intensities.}
We can model the dynamics of interactions between group pairs via a multivariate Hawkes process as follows.

\begin{definition}[Group-Pair Hawkes Intensity Model]
\label{def:group_hawkes}
For each $(i,j) \in \mathcal{P}$, the conditional intensity is defined as
\begin{equation} 
\label{eq:hawkes-intensity}
\lambda_{ij}(t)p
=
\mu_{ij}
+
\sum_{(k,l) \in \mathcal{P}}
\int_{[0,t)}
\phi_{(k,l)\rightarrow(i,j)}(t-s)\,
N_{kl}(ds).
\end{equation}
Here, $\mu_{ij} \ge 0$ denotes the baseline intensity for interactions between groups $i$ and $j$, and
$\phi_{(k,l)\rightarrow(i,j)}$ is an excitation kernel describing how past interactions between groups $(k,l)$ influence future interactions between groups $(i,j)$.
\end{definition}

Within this framework, within-group interactions correspond to group pairs $(i,i)$ and cross-group interactions to pairs $(i,j)$ with $i\neq j$.
We define the aggregated within-group and cross-group counting processes as
\[
N_{\mathrm{w}}(t) = \sum_{i=1}^K N_{ii}(t),
\qquad
N_{\mathrm{c}}(t) = \sum_{1 \le i < j \le K} N_{ij}(t),
\]
with corresponding aggregated intensities
\[
\lambda_{\mathrm{w}}(t) = \sum_{i=1}^K \lambda_{ii}(t),
\qquad
\lambda_{\mathrm{c}}(t) = \sum_{1 \le i < j \le K} \lambda_{ij}(t).
\]

This formulation captures heterogeneous dynamics across group pairs while preserving interpretable global measures of within- and cross-group activity.

\paragraph{Excitation kernel specifications.}
We assume exponential excitation kernels:
\[
\phi_{(k,l)\rightarrow(i,j)}(t)
=
\alpha_{(i,j),(k,l)}(t)\, e^{-\beta t}\,\mathbf{1}_{t \ge 0},
\]
where $\beta > 0$ controls the temporal decay of influence and $\alpha_{(i,j),(k,l)}(t) \ge 0$ denotes the time-dependent excitation strength from group pair $(k,l)$ to group pair $(i,j)$. Exponential kernels are a natural choice due to their tractability and widespread use in modeling temporal effects.

More generally, we allow the excitation strength to depend on time-varying features characterizing the network state:
\[
\alpha_{(i,j),(k,l)}(t)
=
 f\!\left( X_{(i,j),(k,l)}(t) \right),
\]
where $X_{(i,j),(k,l)}(t)$ is a feature representation of the edge pair at time $t$,
and $f$ is a non-negative modulation function.
This allows structural, behavioral, or algorithmic factors to be embedded directly into the excitation dynamics. In particular, $X_{(i,j),(k,l)}(t)$ can be set to the LP score between nodes $k$ and $l$ at time $t$, computed for instance as the cosine similarity between their embeddings, thereby capturing how a LP model may amplify homophily or reinforce inter-group asymmetries over time. A detailed example is provided in Appendix~\ref{app:example}.


\subsection{Empirical vs Instantaneous bias}\label{sec:bias}
 
As presented in Section \ref{sec:back-fairgraph}, homophily is often represented as a function of the number of within- and cross-group edges. In our temporal setting, these quantities naturally correspond to counting processes. 

\begin{definition}[Empirical bias measure]
Let $N_w(t)$ and $N_c(t)$ denote respectively the cumulative number of within-group and cross-group interaction events observed up to time $t$. 
We define the empirical bias measure as
\[
B_{\mathrm{emp}}(t)
=
\frac{N_w(t)}{N_w(t) + N_c(t)}, \quad B_{\mathrm{emp}}(t) \in [0,1].
\]
\end{definition}
$B_{\mathrm{emp}}(t)$ represents the proportion of observed interactions that occur within groups at a given time $t$. Values close to $1$ indicate strong within-group concentration, while values close to $0$ correspond to predominantly cross-group interactions.
In addition, $B_{\mathrm{emp}}(t)$ is cumulative, meaning that as time increases, all past interactions up to time $t$ are equally counted, making the measure increasingly insensitive to recent changes in interaction patterns. As a result, empirical bias summarizes historical imbalance but fails to capture ongoing structural changes or emerging trends in network dynamics.


This limitation motivates assessing bias at the level of the interaction dynamics rather than accumulated outcomes. In a temporal point process framework, each counting process $N_{ij}(t)$ is characterized by a conditional intensity $\lambda_{ij}(t)$.

Therefore, we introduce an instantaneous bias measure based on the aggregated within-group and cross-group interaction intensities.

\begin{definition}[Instantaneous bias]\label{def:inst_bias}
Let $\lambda_{\mathrm{w}}(t)$ and $\lambda_{\mathrm{c}}(t)$ denote the
aggregated within-group and cross-group intensities defined in
Section~\ref{sec:cadre-notation}. We define
\[
B_{\mathrm{inst}}(t)
=
\frac{\lambda_{\mathrm{w}}(t)}
{\lambda_{\mathrm{w}}(t) + \lambda_{\mathrm{c}}(t)}, \quad B_{\mathrm{inst}}(t) \in [0,1]
\]
\end{definition}


$B_{\mathrm{inst}}(t)$ captures the current structural tendency of the system to favor intra-group interactions over inter-group interactions. In particular, $B_{\mathrm{inst}}(t)=0.5$ corresponds to the case where the number of within and cross-group interactions are equal in number; $B_{\mathrm{inst}}(t)\approx 1$ and $B_{\mathrm{inst}}(t)\approx 0$ to strongly homophilic and heterophilic dynamics, respectively.


Unlike $B_{\mathrm{emp}}(t)$, the instantaneous bias reacts immediately to changes in the interaction-generating mechanism, such as variations in excitation parameters, algorithmic exposure, or fairness interventions. In that sense, it reflects what is currently being reinforced by the system, rather than what has accumulated in the past.

One of the main benefits of working with an instantaneous notion of bias is that it naturally separates choice homophily from induced homophily. Based on the Hawkes intensity defined in Equation~\ref{eq:hawkes-intensity}, interaction dynamics can be decomposed into two components: a baseline intensity $\mu_{ij}$ capturing time-independent tendencies for interactions between groups $i$ and $j$, and excitation terms governed by $\alpha_{(k,l)\rightarrow(i,j)}$, which encode endogenous reinforcement mechanisms. By defining bias at the level of instantaneous intensities, $B_{\mathrm{inst}}(t)$ provides a direct approximation of these reinforcement dynamics, allowing us to distinguish persistent underlying trends from time-varying induced effects.

\paragraph{Illustrative Simulation}
We illustrate the difference between empirical and instantaneous bias on a simple synthetic setting, where two groups interact according to Hawkes processes with time-varying reinforcement. The interaction dynamics go through three successive regimes: an initial phase with moderate amplification, an intermediate phase with polarized activity between groups, and a final phase where excitation levels become comparable. Details on parameter estimation are provided in Section~\ref{sec:experiment}, and the simulation setup is described in Appendix~\ref{app:estimation}.

Figure~\ref{fig:bias_hawkes} compares the cumulative empirical bias $B_{\mathrm{emp}}$ with the instantaneous bias computed from the mean interaction intensities. As expected, the empirical bias evolves slowly and remains strongly influenced by past interactions, which smooths out regime changes. In contrast, the instantaneous bias reacts immediately to changes in the underlying interaction mechanism and clearly reflects the transitions between the three phases. Finally, the estimated instantaneous bias closely tracks the oracle bias computed from the true mean intensities, showing that our estimation procedure recovers the underlying dynamics despite stochastic fluctuations and finite observation windows.

\begin{figure}[t]
    \centering
    \includegraphics[width=0.48\textwidth]{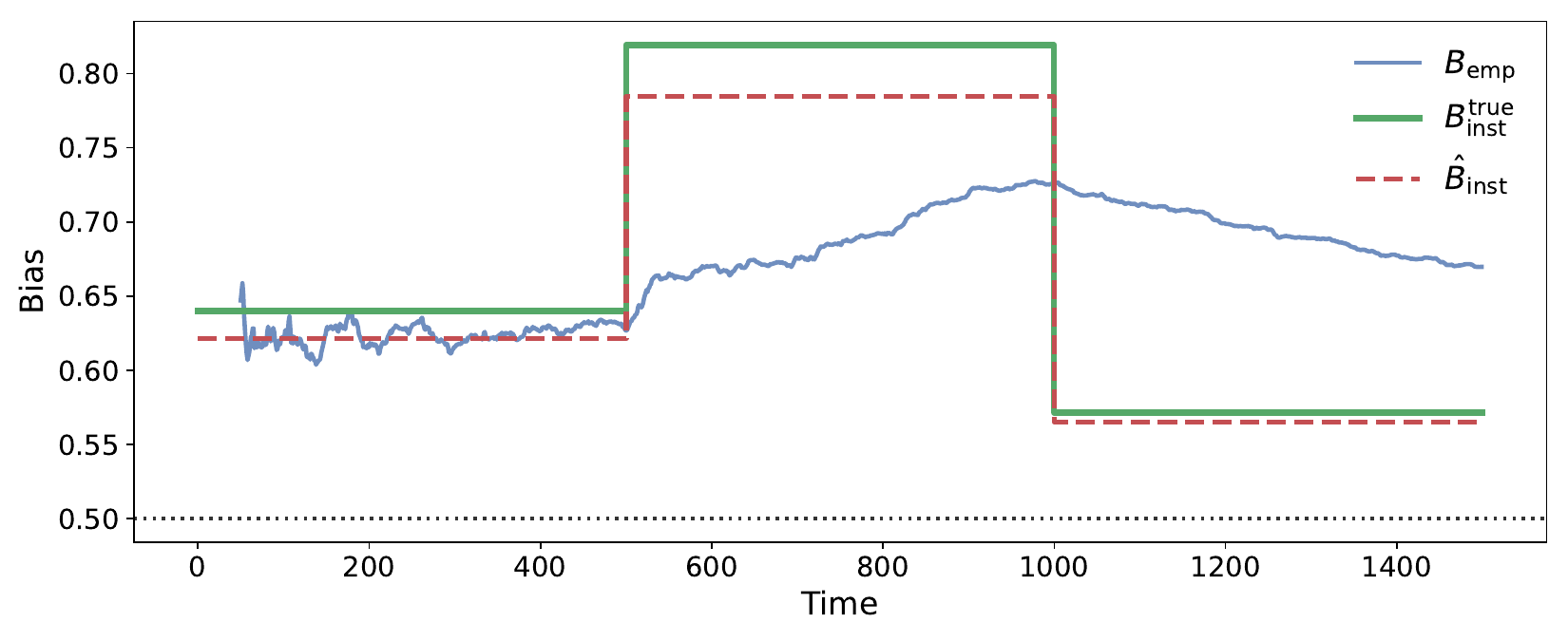}
    \caption{Empirical vs. instantaneous bias from two simulated Hawkes processes across three phases: initial moderate activity ($t<500$), intermediate polarization ($500 \leq t < 1000$), and final alignment of group excitation ($t \geq 1000$).}
    \label{fig:bias_hawkes}
\end{figure}





\section{Theoretical analysis}\label{sec:theorical}
Now, our goal is to characterize how homophily and group-level biases evolve under different interaction and reinforcement regimes, and to identify the conditions under which these dynamics stabilize, persist, or diverge.

\subsection{Mean-Field Approximation.}

We adopt a mean-field perspective \citep{Bacry2015hawkes, delattre2016} in order to obtain a deterministic description of the average interaction dynamics at the group pair level. The central idea is to approximate stochastic interaction increments by their expected behavior, which allows us to average out microscopic fluctuations while preserving the macroscopic structure of reinforcement between groups. This approximation gives rise to a tractable system that allows us to analyze the asymptotic behavior of interaction intensities and derive interpretable stability conditions.

For each group pair $g_{ij}$, with $\mathcal{H}_{t^-}$ the history up to time $t$, the counting process increment satisfies
\[
\mathbb{E}\!\left[d\mathbb{Y}_{g_{ij}}(t)\mid\mathcal{H}_{t^-}\right]
= \lambda_{g_{ij}}(t)\,dt.
\]
We therefore approximate
\[
d\mathbb{Y}_{g_{kl}}(s)
\;\approx\;
\bar{\lambda}_{g_{kl}}(s)\,ds,
\qquad
\bar{\lambda}_{g_{kl}}(t)
:=
\mathbb{E}\!\left[\lambda_{g_{kl}}(t)\right].
\]

Substituting this approximation into the Hawkes dynamics yields a closed system of deterministic integral equations.
\begin{proposition}[Group-level mean-field dynamics]Under a mean-field approximation of the multivariate Hawkes process, the expected intensity for each group pair $g_{ij}$ satisfies
\[
\bar{\lambda}_{g_{ij}}(t)
=
\mu_{g_{ij}}
+
\sum_{(k,l)\in\mathcal{N}_{g_{ij}}}
\int_0^t
\phi_{g_{ij},g_{kl}}(t-s)\,
\bar{\lambda}_{g_{kl}}(s)\,ds.
\]
\end{proposition}
The derivation follows from taking conditional expectations in Equation \ref{eq:hawkes-intensity} and replacing stochastic increments by their mean-field approximation.

We remind that we consider a network with $K$ latent clusters and denote by 
$\lambda_{g_{ij}}(t)$ the intensity of edges between clusters $i$ and $j$ ($1 \le i \le j \le K$).  
Assuming exponential excitation kernels,
\[
\phi_{g_{ij},g_{kl}}(t) = \alpha_{g_{ij},g_{kl}}\, e^{-\beta t}, \quad \beta > 0,
\]
we can define the number of group pairs as 
$ G := \frac{K(K+1)}{2}. $

We introduce the following vector and matrix representations $ \bar{\boldsymbol{\lambda}}(t)=[\bar{\lambda}_{g_{11}}(t), \dots,\bar{\lambda}_{g_{KK}}(t) ]^T \in \mathbb{R}^G$, $\boldsymbol{\mu} = [\mu_{g_{11}}, \dots, \mu_{g_{KK}}]^T \in \mathbb{R}^G$ and $\mathbf{A}=(\alpha_{g_{ij},g_{kl}})_{(i,j),(k,l)} \in \mathbb{R}_+^{G \times G}$.
With these definitions, the mean-field dynamics of the system can be written as:
\begin{equation}\label{eq:meanfield_matrix_compact}
\bar{\boldsymbol{\lambda}}(t) = \boldsymbol{\mu} + \int_0^t \mathbf{A} \, e^{-\beta (t-s)} \, \bar{\boldsymbol{\lambda}}(s) \, ds
.
\end{equation}

This mean-field system captures the average evolution of interaction intensities across group pairs.
It is essential to note that it preserves group-specific differences in baseline tendencies and excitation forces, allowing us to analyze how structural reinforcement mechanisms shape the temporal evolution of homophily and biases.

\subsection{Mean and convergence.}
In this section, we analyze the temporal behavior of the mean intensities defined in Equation \ref{eq:meanfield_matrix_compact}. Our goal is to characterize the stationary mean of and the convergence rate toward this
equilibrium under different spectral regimes of the excitation matrix. This analysis is essential to assess the accuracy of the instantaneous bias \ref{def:inst_bias} and to understand how feedback mechanisms influence long-term homophily.

\paragraph{Stationary regime}
For a multivariate Hawkes process with exponential kernel and fixed excitation matrix $\mathbf A$, stationarity holds if $\rho( \frac{\mathbf A}{\beta}) < 1$ (\cite{BacryMastromatteoMuzy2015}).

In that case, the stationary mean intensity is given by
\[
\bar{\boldsymbol{\lambda}}^*
=
(\mathbf I-\frac{\mathbf A}{\beta})^{-1}\boldsymbol{\mu}.
\]
Additionally, from Definition \ref{def:inst_bias} we have:
\begin{equation}
    B^*_{inst}(t) = \frac{\bar{\lambda}^*_{\mathrm{w}}(t)}
{\bar{\lambda}^*_{\mathrm{w}}(t) + \bar{\lambda}^*_{\mathrm{c}}(t)}
\end{equation}
To assess the reliability of $B^*_{inst}$ as a measure of homophily, it is essential to quantify how quickly $\bar{\boldsymbol{\lambda}}(t)$ converges to its stationary value $\bar{\boldsymbol{\lambda}}^*$.
\begin{proposition}[Exponential convergence rate]\label{prop:convergence}
Assume that $\rho(\mathbf A/\beta) < 1$. 
Then there exist constants $C>0$ and $\kappa>0$ such that
\[
\|\bar{\boldsymbol{\lambda}}(t) - \bar{\boldsymbol{\lambda}}^*\|
\le C e^{-\kappa t}
\|\bar{\boldsymbol{\lambda}}(0) - \bar{\boldsymbol{\lambda}}^*\|.
\]
Moreover, the rate can be chosen such that
\[
\kappa < \beta \big(1 - \rho(\mathbf A/\beta)\big).
\]
\end{proposition}
The proof is given in Appendix \ref{app:proof}. Intuitively, dissipation dominates excitation, ensuring that the mean intensity rapidly stabilizes. When $\rho(\frac{\mathbf A}{\beta})\geq 1$, the system enters critical or supercritical regimes where no stationary means exist; these cases are addressed later in the paper.

\paragraph{Locally stationary regime}
In realistic settings, exposure policies evolve over time, inducing a time-dependent excitation matrix $\mathbf A(t)$.
A natural modeling assumption  is that $\mathbf A(t)$ is piecewise constant on intervals
$[\tau_k, \tau_{k+1})$:
\[
\mathbf A(t) = \mathbf A^{(k)}, \qquad t \in [\tau_k, \tau_{k+1}).
\]
On each interval, the system behaves as a multivariate Hawkes process with fixed parameters, allowing us to reason locally \citep{3680245188524578b5c1c9dc5525b52a}, which is particularly relevant for LP analysis. If convergence to equilibrium is rapid relative to the rate at which policies change, the system operates near its local equilibrium $\bar{\boldsymbol{\lambda}}^{*(k)}$, and  bias measures calculated from this equilibrium are meaningful. However, if convergence is slow, persistent transient imbalances may arise, leading to systematic exposure disparities that are not captured by steady-state analyses. Thus, the convergence rate $\kappa_k$ determines whether measures calculated from the local equilibrium accurately reflect the system's behavior, or whether transient dynamics dominate.




\begin{figure*}[tp]
\begin{subcolumns}[0.48\textwidth]
\includegraphics[width=\subcolumnwidth]{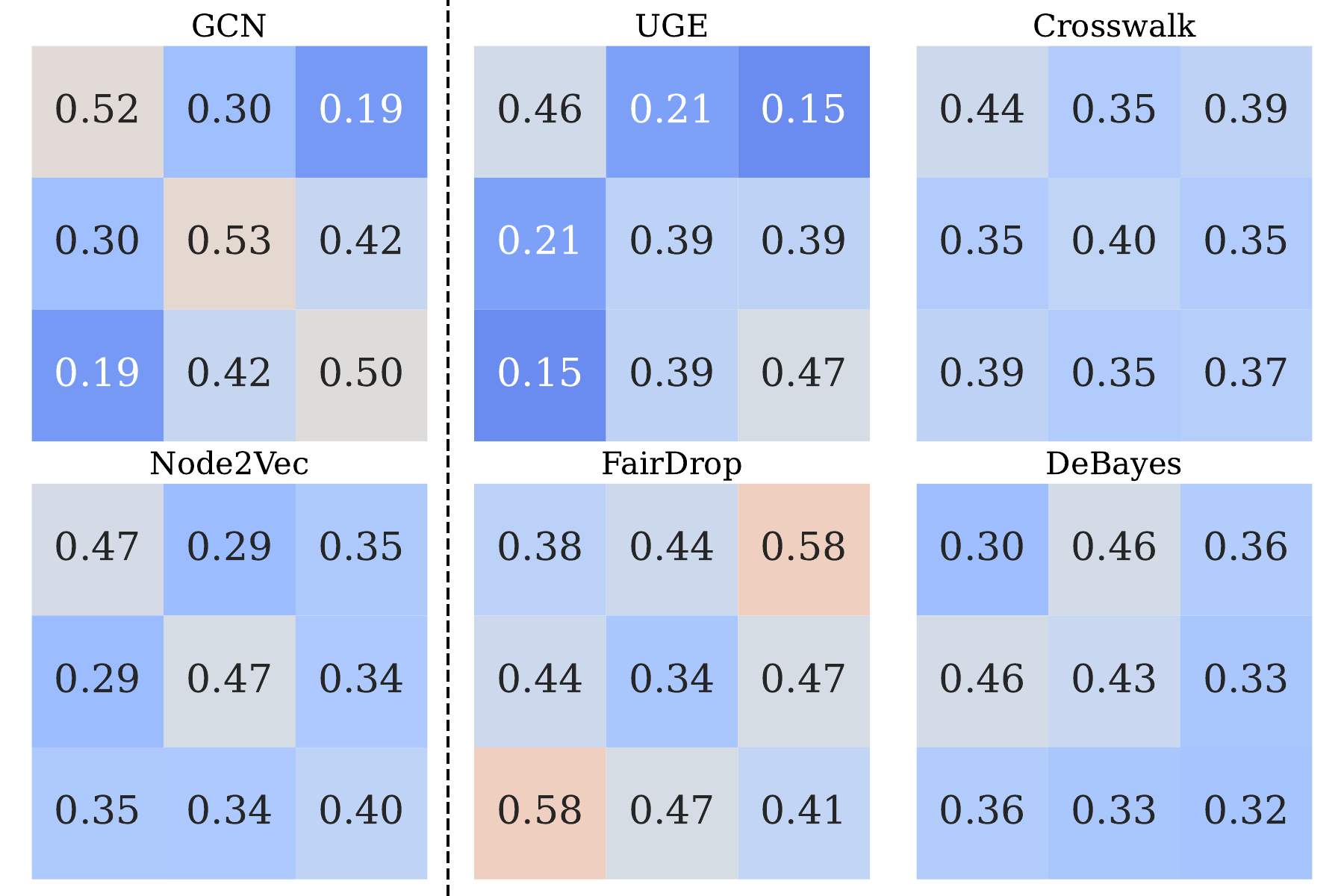}
\nextsubcolumn[0.47\textwidth]
\includegraphics[width=0.85\subcolumnwidth]{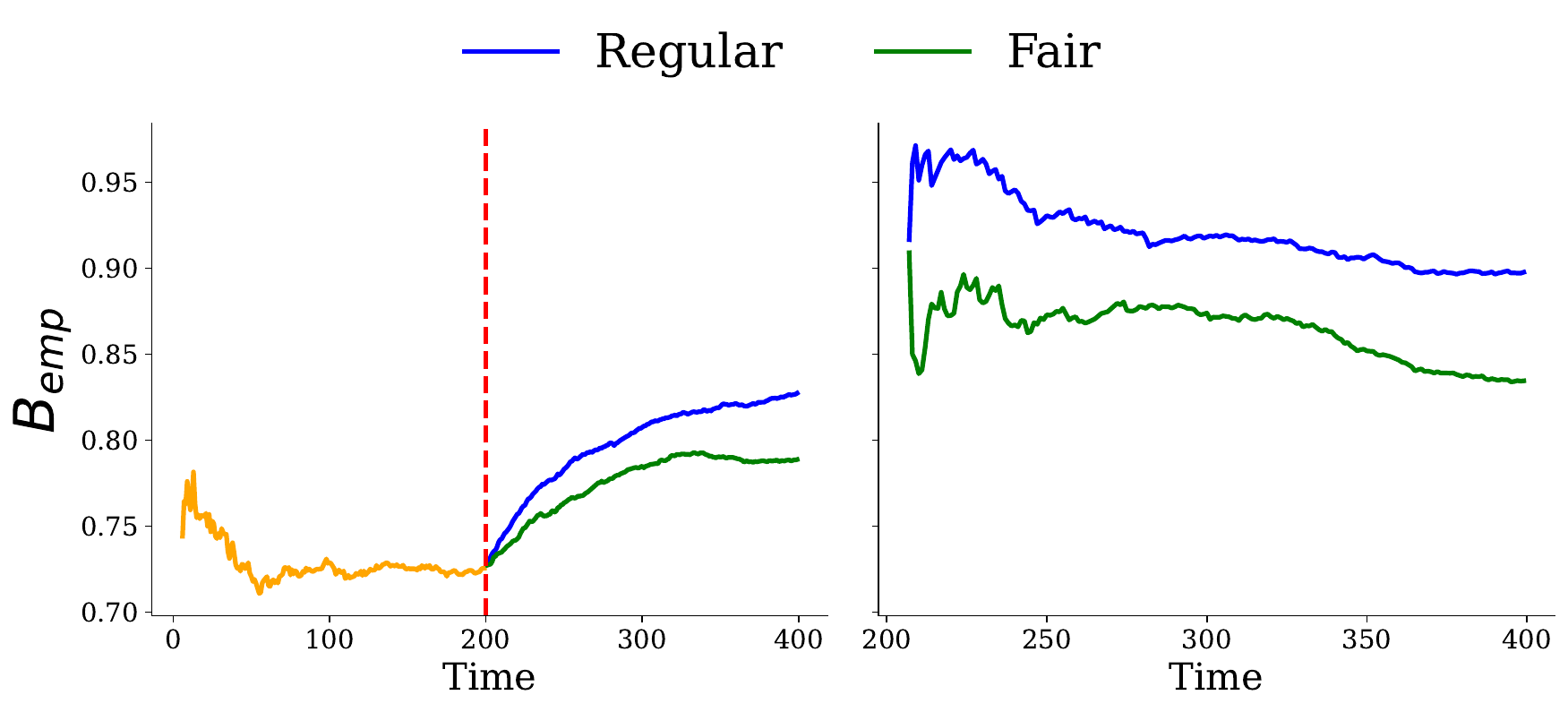}\\
\includegraphics[width=0.85\subcolumnwidth]{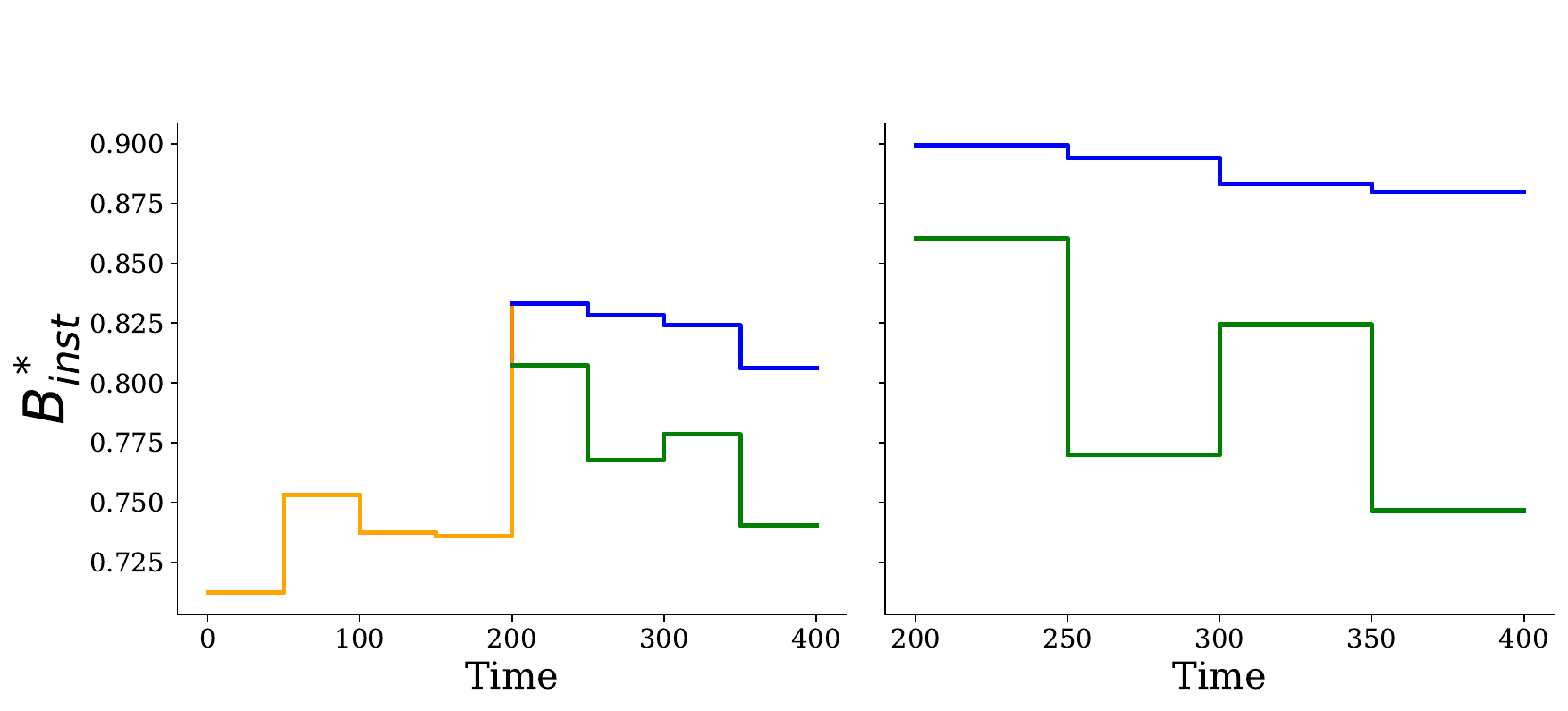}
\end{subcolumns}
\caption{Results on synthetic data. Left: Mean self-excitation coefficients $\alpha_{(i,j)}$ estimated for each model, where entry $(i,j)$ represents the estimated self-excitation of the interaction process between groups $i$ and $j$. Right-top: $B_{emp}(t)$ evolution: full timeline v.s. post-intervention phase. Right-bottom:$B^*_{inst}(t)$ evolution: full timeline v.s. only LP-intervention phase.}
\label{fig:results-synthetic}
\end{figure*}
\begin{proposition}[Exponential convergence rate - locally stationary case]
\label{prop:convergence-local}
Consider an interval $[\tau_k, \tau_{k+1})$ on which the excitation matrix is constant: $\mathbf{A}(t) = \mathbf{A}^{(k)}$. 

Assume that $\rho(\mathbf{A}^{(k)}/\beta) < 1$. Then there exist constants $C_k > 0$ and $\kappa_k > 0$ such that, for all $t \in [\tau_k, \tau_{k+1})$,
\[
\|\bar{\boldsymbol{\lambda}}(t) - \bar{\boldsymbol{\lambda}}^{*(k)}\|
\leq C_k e^{-\kappa_k (t - \tau_k)}
\|\bar{\boldsymbol{\lambda}}(\tau_k) - \bar{\boldsymbol{\lambda}}^{*(k)}\|,
\]
where $\bar{\boldsymbol{\lambda}}^{*(k)} = (\mathbf{I} - \mathbf{A}^{(k)}/\beta)^{-1}\boldsymbol{\mu}$ is the local equilibrium. Moreover, the rate can be chosen such that
\[
\kappa_k < \beta \big(1 - \rho(\mathbf{A}^{(k)}/\beta)\big).
\]
Consequently, the instantaneous bias $B_{\text{inst}}(t)$ converges to a well-defined limit within $O(\kappa_k^{-1})$ on each interval. 
\end{proposition}

\begin{proof}
    The first part follows from applying Proposition \ref{prop:convergence} on each interval $[\tau_k, \tau_{k+1})$. The convergence of $B_{\text{inst}}(t)$ follows from the fact that it is a continuous function of the intensities $\bar{\mathbf{\lambda}}(t)$.
\end{proof}

\paragraph{Implications for algorithmic interventions}
Proposition \ref{prop:convergence-local} provides theoretical insight to interpret the temporal effects of algorithmic interventions on bias. 

Interventions modify exposure dynamics over time, yielding a sequence of excitation matrices $\mathbf{A}^{(k)}$ across policy windows. The spectral radius $\rho(\mathbf{A}^{(k)}/\beta)$, determines whether the dynamics admit a stable regime and how fast transients decay. When $\rho(\mathbf{A}^{(k)}/\beta)<1$,  bias indicators can be meaningfully interpreted once the system has reached its local equilibrium. As the system approaches the critical regime, convergence becomes slow and short-term observations are dominated by transient effects; beyond this threshold, no stable regime is expected, and bias measures may fluctuate indefinitely or diverge, making reliable evaluation difficult. 

Additionally this result implies that stability does not imply bias reduction. For instance, standard LP models tend to reinforce within-group interactions, while fairness-oriented interventions promote cross-group exposure; although they act in opposite directions in terms of bias at equilibrium, both increase the overall level of excitation in the system, pushing the dynamics closer to critical regimes. As a result, interventions designed to reduce bias may require longer observation windows before their long-term effects can be reliably assessed. This also points to a limitation of purely static fairness evaluations in LP: when predictions are fed back into the interaction process, fairness properties become inherently dynamic and may not be captured by metrics computed on a fixed graph.

\section{Experiments}\label{sec:experiment}
The research questions motivating our experiments are:
\begin{itemize}
    \item \textbf{RQ1:} Is the $B_{\text{inst}}$ measure appropriate for capturing the dynamic evolution of homophily in a network?
    \item \textbf{RQ2:} How do different approaches to fairness influence the temporal dynamics of homophily?
    \item \textbf{RQ3:} What is the impact of the retraining frequency of prediction models on the evolution of homophily?
\end{itemize}

We conduct two types of experiments. First, we use simulated networks to validate our $B^*_{inst}$ measure and evaluate various fair LP models, with and without model retraining. This allows us to answer RQ1--RQ3. 

\paragraph{Parameter estimation.} First, we assume that cross-pair excitations are negligible, i.e.\ $\alpha_{(i,j),(k,l)} = 0$ for $(i,j) \neq (k,l)$. This diagonal assumption is motivated by identifiability constraints given the finite observation windows available per pair. Under this simplification, each $\alpha_{(i,j)}$ is estimated via maximum likelihood, maximizing the log-likelihood of the observed event sequence, and the spectral radius reduces to $\rho(A/\beta) = \max_{(i,j)} \alpha_{(i,j)}/\beta$, recovering a closed-form stationary intensity $\bar{\lambda}^*_{ij} = \mu_{ij}/(1 - \alpha_{ij}/\beta)$. The $\beta$ is fixed at $1$ and $\mu$ is estimated via its closed-form MLE $\hat{\mu} = N/T$, where $N$ is the number of observed events over $[0, T]$, following standard practice in the Hawkes process literature~\cite{daley2003introduction}.

\subsection{Simulation Protocol}\label{sec:protocol}

For the simulations, we draw inspiration from stochastic block models (SBMs) to model homophily of choice. We first simulate a pre-network of $N$ nodes partitioned into sensitive groups, where each node is associated with an embedding drawn from a latent cluster structure representing latent characteristics, independent of the sensitive group attribute. At each timestep, active nodes are sampled at random, and three candidate nodes are proposed based on cosine similarity in the embedding space. Link formation probabilities depend on the candidates’ ranking, the SBM group structure, and node popularity. This construction ensures that the initial phase reflects homophily of choice only, without any algorithmic intervention.

After this phase, a LP model is trained on the pre-network to generate more informative embeddings. Depending on the retraining policy, the model can be updated periodically or remain fixed, allowing us to study how algorithmic recommendations, and their potential adaptation over time, interact with the initial network structure. 
This framework allows simulation of the evolution of a social network under the combined influence of its initial structure and recommendations. Details are provided in Appendix \ref{app:simulation_param}.

\subsection{$B^*_{inst}$ for Dynamic Homophily (RQ1)}

First, we highlight the interest of $B^*_{inst}$ compared to the traditional model $B_{cum}$. After generating the pre-network, we run a fair LP model and a regular one (retrained several times) starting from $t=200$. Figure~\ref{fig:results-synthetic} shows that $B_{cum}$ fails to capture the rapid changing homophilic dynamics induced by the models and smooths out the changes over time. In contrast, $B^*_{inst}$ is largely insensitive to whether it is measured throughout the timeline or only during the LP activity period, confirming its reliability as a measure of homophilic dynamics. For all calculations, $\alpha$ is estimated in the retraining windows, and $\mu$ is estimated from the pre-network phase; when limited to the second phase, we set $\mu_w = \mu_c$ to isolate homophilic dynamics.
%
%
We further observe that the fair LP model is less polarizing than the unfair model, which motivates the analysis in the next section.

\subsection{Fairness Interventions and Retraining Dynamics (RQ2 and RQ3)}
We compare several fairness-aware LP methods from the literature on a synthetic graph with three sensitive groups. These include \textbf{CrossWalk} \cite{Khajehnejad_Khajehnejad_Babaei_Gummadi_Weller_Mirzasoleiman_2022}, \textbf{UGE \cite{Wang_2022}}, \textbf{DeBayes} \cite{buyl2020debayes}, and \textbf{FairDrop} \cite{Spinelli_2022}. We also include \textbf{Node2Vec}   \cite{grover2016node2vecscalablefeaturelearning} and a GCN as regular baselines for LP. More details on LP models in Appendix~\ref{app:simulation_param}. 

Table~\ref{tab:3_cluster} and Figure~\ref{fig:results-synthetic} jointly validate $B^*_{inst}$ as a reliable measure of homophilic dynamics. A key distinction motivates our approach: $\Delta DP$ is a model-dependent fairness metric requiring access to 
the model's predictions, whereas $B^*_{inst}$ measures homophily directly from the graph's interaction dynamics, without knowledge of the underlying model. Despite this difference, both measures are consistent across all models: higher $B^*_{inst}$ systematically corresponds to larger $\Delta DP$, confirming that homophilic dynamics observed at the graph level are 
a reliable reflection of model-level fairness outcomes. The Hawkes decomposition further allows us to disentangle the origins of this homophily. The $\mu$ matrix (Appendix~\ref{app:simulation_param}), estimated on the pre-network phase, reflects the baseline structural homophily already present in the graph before any algorithmic intervention. The $\alpha$ matrices then act as a proxy for the model's dynamic influence on the graph: GCN exhibits strongly concentrated diagonal coefficients, indicating strong algorithmic reinforcement of within-group interactions, while fairness-aware models such as FairDrop and DeBayes produce more uniform matrices, reflecting weaker group-level excitation. Notably, FairDrop exhibits slightly elevated off-diagonal $\alpha$ coefficients, yielding a higher spectral radius, which could potentially lead to a long-term explosion of cross-group interactions and warrants monitoring over extended time horizons.

Regarding retraining, the significant increase in $B^*_{inst}$ observed for GCN confirms that, without fairness constraints, periodic model updates amplify homophily in the network.

\begin{table}[tp]
\centering
\caption{AUC, homophily $B^*_{inst}$ and $\Delta DP$ averaged over 10 runs (mean $\pm_{\text{std}}$). \textit{Without retrain}: fixed embedding. \textit{With retrain}: embedding recomputed in middle timeline.}
\label{tab:3_cluster}
\resizebox{0.5\textwidth}{!}{
\begin{tabular}{lcccc|cc}
\toprule
& \multicolumn{4}{c|}{\textit{Without retrain}} & \multicolumn{2}{c}{\textit{With retrain}} \\
\cmidrule(lr){2-5} \cmidrule(lr){6-7}
Model & AUC & AUC 1 & $B^*_{inst}$ & $\Delta DP$ & $B^*_{inst}$ & $\Delta DP$ \\
\midrule
GCN   & $0.84_{\pm 0.01}$ & $0.97_{\pm 0.00}$ & $0.75_{\pm 0.02}$ & $0.37_{\pm 0.04}$ & $0.80_{\pm 0.03}$ & $0.43_{\pm 0.03}$ \\
Crosswalk & $0.84_{\pm 0.01}$ & $0.96_{\pm 0.01}$ & $0.73_{\pm 0.01}$ & $0.35_{\pm 0.02}$ & $0.81_{\pm 0.02}$ & $0.43_{\pm 0.04}$ \\
DeBayes    & $0.85_{\pm 0.01}$ & $0.88_{\pm 0.00}$ & $0.64_{\pm 0.01}$ & $0.18_{\pm 0.02}$ & $0.63_{\pm 0.02}$ & $0.18_{\pm 0.01}$ \\
UGE  & $0.86_{\pm 0.01}$ & $0.88_{\pm 0.01}$ & $0.68_{\pm 0.01}$ & $0.24_{\pm 0.03}$ & $0.71_{\pm 0.02}$ & $0.31_{\pm 0.03}$ \\
Node2Vec  & $0.85_{\pm 0.01}$ & $0.94_{\pm 0.01}$ & $0.63_{\pm 0.02}$ & $0.20_{\pm 0.02}$ & $0.65_{\pm 0.02}$ & $0.20_{\pm 0.02}$ \\
FairDrop      & $0.84_{\pm 0.00}$ & $0.84_{\pm 0.00}$ & $0.63_{\pm 0.01}$ & $0.17_{\pm 0.01}$ & $0.61_{\pm 0.02}$ & $0.17_{\pm 0.01}$ \\

\bottomrule
\end{tabular}
}
\end{table}


\subsection{Real-World Networks}
We consider a subset of the Twitter/X 2021 dataset proposed by \cite{Pournaki_Gaisbauer_Olbrich_2025}, focusing on interactions related to political content. This period is of particular interest as it coincides with an election year, during which polarization dynamics are known to be amplified. To define the sensitive attribute, we infer users' political orientation from hashtag usage, distinguishing coarse-grained ideological groups. The labeling protocol is described in Appendix \ref{app:real-world}.
\begin{figure}[tp]
    \centering
    \includegraphics[width=0.48\textwidth]{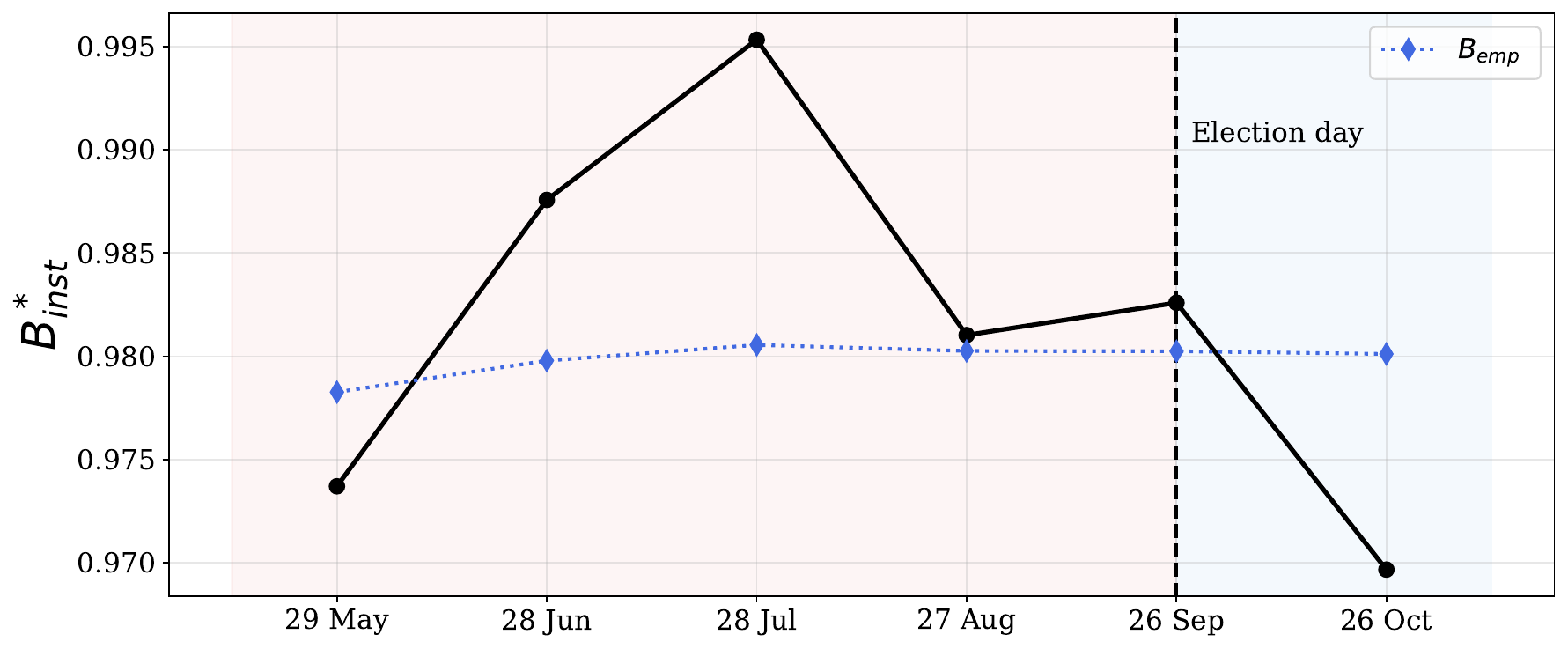}
    \caption{$B^*_{\mathrm{inst}}$  around the German federal election of 2021.}
    \label{fig:election}
\end{figure}

We observe a marked peak of $B^{\mathrm{inst}}$ (Figure \ref{fig:election}) in July~2021, well before the election, consistent with a transient amplification of within-group self-excitation. The high baseline level of $B^{\mathrm{inst}}$ throughout the window indicates that the network is already strongly polarized outside the peak, and that our measure mainly captures variations in reinforcement on top of an existing polarized structure. We do not claim a causal link, but note that this period coincides with heightened public attention during the campaign.

\section{Conclusion and perspectives}
We analyzed how link prediction models, when coupled with network dynamics, shape the evolution of homophily and group-level disparities over time. Modeling interactions with Hawkes processes allows us to separate baseline interaction tendencies from reinforcement effects driven by past interactions and algorithmic exposure, and to define an instantaneous bias measure that captures ongoing structural tendencies beyond cumulative homophily indicators.

Our theoretical results clarify when such bias measures are meaningful in terms of stability and convergence of the induced dynamics, and experiments highlight contrasted homophily trajectories across link prediction strategies and retraining policies. Overall, this suggests that fairness assessments based on static snapshots can be misleading when predictions actively shape future interactions. Relaxing simplifying assumptions on the excitation structure would allow finer-grained modeling of cross-group reinforcement effects. More broadly, our results suggest that designing and evaluating graph learning methods requires explicit attention to their dynamic, feedback-driven impact on network structure.

\newpage
\bibliography{uai2026-template}

\newpage 
\onecolumn

\title{How Predicted Links Influence Network Evolution:\\Disentangling Choice and Algorithmic Feedback in Dynamic Graphs\\(Supplementary Material)}
\maketitle

\appendix

\section{Appendix: Two-Group Illustrative Example}\label{app:example}

We consider a professional network with two demographic groups: group 1 (e.g., men (M)) and group 2 (e.g., women (W)). Interactions in this context correspond to repeated professional exchanges such as endorsements, messages, or recommendations. The mark space reduces to three group pairs: $(1,1)$ (M--M), $(2,2)$ (F--F), and $(1,2)$ (M--F).

\paragraph{Model parameters.}
Baseline intensities reflect a moderate choice homophily: men interact more frequently among themselves ($\mu_{11} = 0.8$), women less so ($\mu_{22} = 0.5$), and cross-gender interactions are rare in the absence of algorithmic intervention ($\mu_{12} = 0.2$). The decay parameter is set to $\beta = 1$.

The non-zero entries of the excitation matrix $\mathbf{A}$ are reported in Table~\ref{tab:phi}. We adopt the sparse neighbourhood structure introduced above: only group pairs sharing a common group index have non-zero excitation. Self-excitation is strongest for M--M interactions ($\alpha_{(1,1),(1,1)} = 0.60$), reflecting the strong tendency of men to engage with other men once a connection is established. Cross-gender self-excitation is weaker ($\alpha_{(1,2),(1,2)} = 0.20$), capturing the limited tendency of cross-gender interactions to self-reinforce without external intervention.

\begin{table}[h]
\centering
\caption{Non-zero entries of the excitation matrix $\mathbf{A}$, 
i.e., excitation strengths $\alpha_{(i,j),(k,l)}$ in the two-group example, with $\alpha_0 = 0.3$ and $\gamma_{\mathrm{rec}} = 0.5$.}
\label{tab:phi}
\renewcommand{\arraystretch}{1.3}
\begin{tabular}{llp{7cm}}
\toprule
\textbf{Kernel} & \textbf{Value} & \textbf{Interpretation} \\
\midrule
$\alpha_{(1,1),(1,1)}$ & $0.60$ & M--M self-excitation: each M--M interaction reinforces future ones \\
$\alpha_{(2,2),(2,2)}$ & $0.40$ & F--F self-excitation: weaker than M--M due to smaller group activity \\
$\alpha_{(1,2),(1,2)}$ & $0.20$ & Cross-gender self-excitation: weak without algorithmic boost \\
$\alpha_{(1,2),(1,1)}$ & $0.10$ & Cross-gender interactions mildly stimulate M--M activity \\
$\alpha_{(1,2),(2,2)}$ & $0.10$ & Cross-gender interactions mildly stimulate F--F activity \\
$\alpha_{(1,1),(1,2)}$ & $0.05$ & M--M interactions weakly stimulate cross-gender ones \\
$\alpha_{(2,2),(1,2)}$ & $0.05$ & F--F interactions weakly stimulate cross-gender ones \\
\bottomrule
\end{tabular}
\end{table}

The resulting matrix $\mathbf{A}$, whose entry
$[\mathbf{A}]_{(i,j),(k,l)} = \alpha_{(i,j),(k,l)}$,
ordered as $(1,1), (2,2), (1,2)$, takes the form:
\[
\boldsymbol{A} =
\begin{bmatrix}
\textcolor{men}{0.60} & 0 & \textcolor{cross}{0.10} \\
0 & \textcolor{women}{0.40} & \textcolor{cross}{0.10} \\
\textcolor{men}{0.05} & \textcolor{women}{0.05} & \textcolor{cross}{0.20}
\end{bmatrix},
\]
where rows correspond to target group pairs and columns to source group pairs. {\color{men} Blue} entries involve M--M dynamics, {\color{women} red} entries F--F dynamics, and {\color{cross} green} entries cross-gender dynamics.
The corresponding kernel matrix is then $\boldsymbol{\Phi}(t) = \mathbf{A}\, e^{-\beta t}$.

\paragraph{Algorithmic scenarios}
We compare a standard scenario that amplifies within-group connections and a 
fairness-aware scenario that boosts cross-group exposure by increasing $E_{12}$. Both satisfy $\rho(\mathbf{A}/\beta) < 1$, ensuring convergence to a well-defined equilibrium. Figure~\ref{app:fig-network} illustrates the resulting network dynamics and equilibrium intensities.

\begin{figure}[h!]
    \centering
    \begin{minipage}[c]{0.62\linewidth}
        \includegraphics[width=\linewidth]{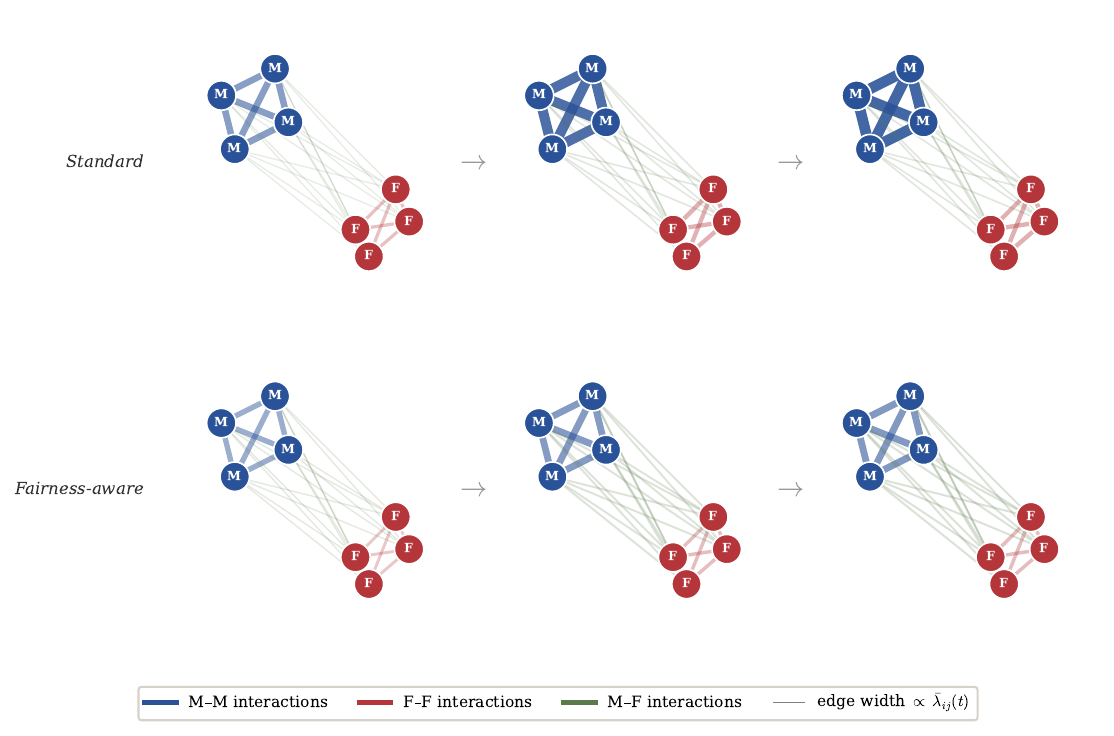}
    \end{minipage}
    \hfill
    \begin{minipage}[c]{0.34\linewidth}
        \includegraphics[width=\linewidth]{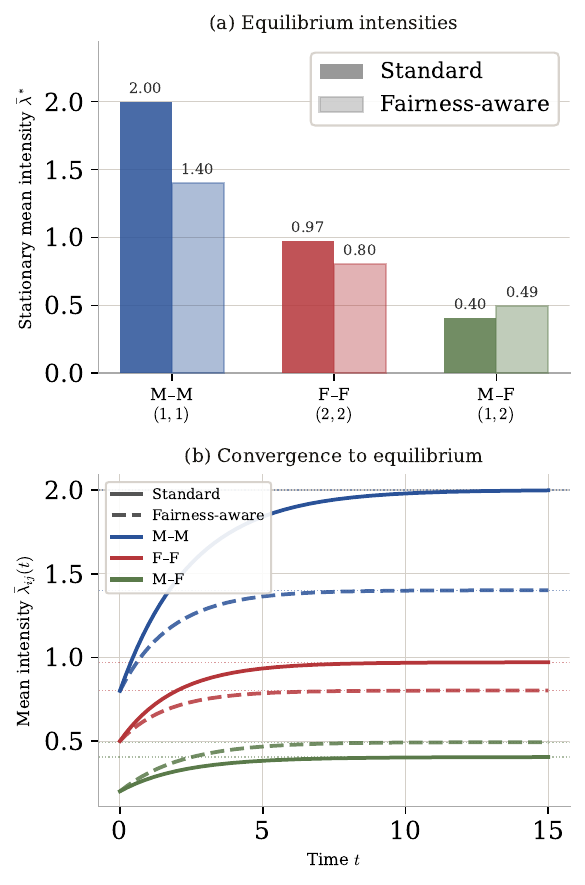}
    \end{minipage}
    \caption{Network evolution and equilibrium intensities under two algorithmic 
    scenarios on a two-group professional network (men M, women F). 
    \emph{Left:} snapshots of the interaction network at three time points, 
    with edge width proportional to the mean-field intensity $\bar{\lambda}_{ij}(t)$. 
    Under the standard scenario (top row, $E_{12} = 0.05$), within-group interactions 
    (blue: M--M, red: F--F) progressively dominate, reflecting the amplification 
    of baseline homophily through self-excitation. 
    Under the fairness-aware scenario (bottom row, $E_{12} = 0.90$), 
    cross-group interactions (green: M--F) are reinforced instead. 
    \emph{Right:} (a) stationary mean intensities $\bar{\lambda}^*$ at equilibrium; 
    (b) mean-field convergence trajectories $\bar{\lambda}_{ij}(t)$. 
    The fairness-aware intervention reduces $\bar{\lambda}^*_{11}$ and 
    increases $\bar{\lambda}^*_{12}$ relative to the standard scenario
    }
    \label{app:fig-network}
\end{figure}





\section{Details on the illustrative simulation}
\label{app:estimation}

\paragraph{Simulation setup.}
We simulate two independent point processes corresponding to within-group interaction streams for groups $w$ and $c$, using Hawkes processes with exponential kernel. Events are generated with Ogata’s thinning algorithm over a time horizon $T=1500$ and decay rate $\beta=1$. Baseline intensities are set to $\mu_w=0.8$ and $\mu_c=0.6$. 

To emulate regime changes in interaction dynamics, excitation parameters are chosen to be piecewise constant with two change points at $t=500$ and $t=1000$. Specifically, for group $w$, the self-excitation coefficient increases in the intermediate phase before decreasing in the final phase, while for group $c$ it initially decreases and then increases, leading to three successive regimes with contrasted levels of endogenous reinforcement. This construction is meant to reflect qualitative shifts in interaction dynamics rather than to model a specific real-world system.
\paragraph{Parameter Estimation}
The Table \ref{tab:values} shows that the estimation successfully recovers both the baseline intensities $\mu$ and the piecewise self-excitation coefficients $\alpha$, with estimates closely matching the true values for each regime.

\begin{table}[h!]
\centering
\caption{Self-excitation coefficients $\alpha$ and baseline intensities $\mu$ for each regime and their estimates.}
\begin{tabular}{c ccc ccc}
\toprule
Group & \multicolumn{3}{c}{True $\alpha$ values} & \multicolumn{3}{c}{Estimated values by window} \\
\cmidrule(lr){2-4} \cmidrule(lr){5-7}
       & Regime 1 & Regime 2 & Regime 3 & Window 1 & Window 2 & Window 3 \\
\midrule
w      & 0.40 & 0.75 & 0.50 & 0.435 & 0.733 & 0.543 \\
c      & 0.20 & 0.15 & 0.50 & 0.240 & 0.198 & 0.490 \\
\bottomrule
\end{tabular}
\label{tab:values}

\vspace{0.2cm}
\begin{tabular}{ccc}
\toprule
Parameter & Estimated value & True value \\
\midrule
$\mu_w$ & 0.717 & 0.800 \\
$\mu_c$ & 0.600 & 0.600 \\
\bottomrule
\end{tabular}
\end{table}

\section{Technical Proof} \label{app:proof}

\begin{proof}
We aim to show that if $\rho(\frac{\mathbf A}{\beta})<1$, the error 
$e(t) = \bar{\boldsymbol{\lambda}}(t) - \bar{\boldsymbol{\lambda}}^*$ decays exponentially and 
satisfies 
\[
\| e(t) \| \le C e^{-\kappa t} \| e(0)\|,
\]
with a rate $\kappa < \beta(1-\rho(\frac{\mathbf A}{\beta}))$. \\

The proof proceeds in two steps: we first reformulate the mean-field dynamics as a linear ODE, then analyze the error dynamics using spectral properties of the excitation matrix.



\paragraph{Step 1: Differential formulation.} 

Starting from Equation \ref{eq:meanfield_matrix_compact}, we differentiate with respect to time using the Leibniz integral rule. Since $\mu$ is constant and the integrand has the form $f(t,s)=Ae^{-\beta(t-s)}\bar{\lambda}(s)$, we obtain 
\begin{align*}
\dot{\bar{\lambda}}(t) &= Ae^{0}\bar{\lambda}(t) + \int_0^t (-\beta)Ae^{-\beta(t-s)}\bar{\lambda}(s)\,ds = A\bar{\lambda}(t) - \beta\left[\bar{\lambda}(t) - \mu\right].
\end{align*}
Rearranging Equation \ref{eq:meanfield_matrix_compact}, we recognize that $\int_0^t Ae^{-\beta(t-s)}\bar{\lambda}(s)\,ds = \bar{\lambda}(t) - \mu$. Substituting this into the above expression yields the linear ODE
\begin{equation}\label{eq:ode}
\dot{\bar{\lambda}}(t) = A\bar{\lambda}(t) - \beta[\bar{\lambda}(t) - \mu] = (A - \beta I)\bar{\lambda}(t) + \beta\mu.
\end{equation}

At equilibrium, $\dot{\bar{\lambda}}^*=0$ implies $(A - \beta I)\dot{\bar{\lambda}}^* + \beta\mu$, which give the well-known stationary mean:
\begin{equation}\label{eq:stationary-mean}
    \bar\lambda^*=(I-A/\beta)^{-1}\mu.
\end{equation}
Note that this invertibility of $(I-A/\beta)$ is guaranteed by the assumption $\rho(A/\beta)<1$, since $||(A/\beta)^n||\rightarrow 0$ ensures the Neumann series $\sum_{n=0}^\infty (A/\beta)^n$ converges.

\paragraph{Step 2: Exponential decay of the error.} Starting from
\begin{equation} \label{eq:deriv_hawkes}
\dot{\bar{\lambda}}(t) 
= (A - \beta I) \, \bar{\lambda}(t) + \beta \, \mu,
\end{equation}
we can define the error $e(t) = \bar{\lambda}(t)- \bar{\lambda}^*$. Differentiating and using Equations \ref{eq:ode} and \ref{eq:stationary-mean}, we obtain 

\begin{align*}
    \dot{e}(t) & = \dot{\bar{\lambda}}- \dot{\bar{\lambda}}^* = \dot{\bar{\lambda}} \\
    & = (A - \beta I) \bar{\lambda}(t) + \beta \mu \\
    &= (A - \beta I)(e(t) + \bar{\lambda}^*) + \beta \mu\\
    &=(A- \beta I)e(t) + \underbrace{(A - \beta I)\bar{\lambda}^* + \beta \mu}_{=0}.
\end{align*}

Thus, the error statisfies the homogeneous linear system:

\begin{equation}
\dot e(t) = (A - \beta I)\, e(t), \qquad e(0)=\bar\lambda(0)-\bar\lambda^*.
\end{equation}

The solution is $e(t) = \text{exp}[{(A-\beta I)t}] e(0)$. To bound the matrix exponential, we analyse the spectrum of $A-\beta I$. If $v$ is an eigenvector of $A/\beta$ with eigenvalues $\nu$, then:

$$(A-\beta I)v=Av-\beta v=\beta \nu v-\beta v= \beta(\nu-1)v.$$

Hence the eigenvalues of $A-\beta I$ are $\{\beta(\nu_i-1)\}$ where $\{\nu_i\}$ are eigenvalues of $A/\beta$. Since $\rho(A/\beta)<1$, we have $|\nu_i|\leq\rho(A/\beta)<1$ for all $i$, which implies: 
$$\Re[\beta(\nu_i-1)]=\beta(\Re(\nu_i)-1)<0,$$
making $A-\beta I$ a Hurwitz matrix. By standard results on linear ODEs, there exist constants $C,\kappa>0$ such that 
$$
||\text{exp}[(A-\beta I)t]||\leq Ce^{-\kappa t}, \quad \forall t\geq 0,
$$
with $\kappa<-\text{max}_i\Re[\beta(\nu_i-1)]=\beta(1-\text{max}_i\Re(\nu_i)).$
When $A/\beta$ is nonnegative (as is natural in our setting), the Perron-Frobenius theorem guarantees that $\rho(A/\beta)=\text{max}_i\Re(\nu_i)$, yielding:
$$\kappa<\beta(1-\rho(A/\beta)).$$
Consequently,
\begin{equation}
\|\bar{\boldsymbol{\lambda}}(t) - \bar{\boldsymbol{\lambda}}^*\|
\le C e^{-\beta(1-\rho(A/\beta)) t}
\|\bar{\boldsymbol{\lambda}}(0) - \bar{\boldsymbol{\lambda}}^*\|,
\end{equation}
which can be written as $||e(t)||\leq Ce^{-\kappa t}||e(0)||$, completing the proof.

\end{proof}

Figure~\ref{fig:proposition3} illustrates this result empirically: 
across three successive intervals with different excitation matrices, 
the normalized distance to the local equilibrium consistently stays 
below the theoretical exponential bound, confirming that the 
convergence rate $\kappa_k$ is well captured by $\beta(1-\rho(\mathbf{A}^{(k)}/\beta))$.
\begin{figure}[h]
    \centering
    \includegraphics[width=0.83\linewidth]{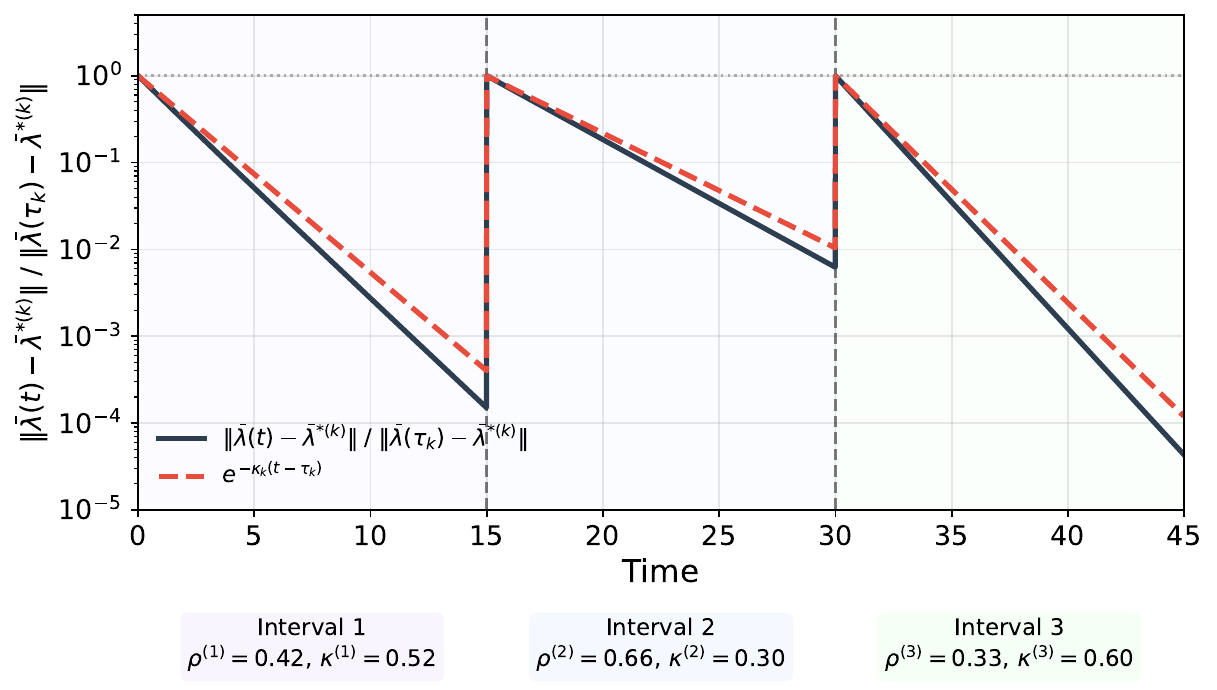}
    \caption{Empirical verification of Proposition~\ref{prop:convergence}.
    The normalized distance $\|\bar{\boldsymbol{\lambda}}(t) - \bar{\boldsymbol{\lambda}}^{*(k)}\|
    / \|\bar{\boldsymbol{\lambda}}(\tau_k) - \bar{\boldsymbol{\lambda}}^{*(k)}\|$
    (solid line) remains below the exponential bound $e^{-\kappa_k(t - \tau_k)}$
    (dashed line) on each interval $[\tau_k, \tau_{k+1})$, with
    $\kappa_k = 0.9\,\beta(1 - \rho(\mathbf{A}^{(k)}/\beta))$.
    The excitation matrix $\mathbf{A}^{(k)}$ changes twice, inducing
    a new local equilibrium $\bar{\boldsymbol{\lambda}}^{*(k)}$ and a
    fresh convergence phase at each transition.}
    \label{fig:proposition3}
\end{figure}

\section{Additional details on Experiments}
\subsection{Simulation on Synthetic Data}\label{app:simulation_param}
\paragraph{Simulation parameter and procedure} We present in Table \ref{tab:simulation_params} the hyper-parameters present in our simulation and their role. In addition, Algorithm \ref{alg:pre_network} details the step to generate the temporal network used as a starting point for our experiments. 
\begin{table}[h]
\centering
\caption{Hyperparameters used in the simulation.}
\label{tab:simulation_params}
\footnotesize
\begin{tabular}{lll}
\hline
\textbf{Concept} & \textbf{Parameter} & \textbf{Description} \\
\hline
Network size & $N = 300$ & Total number of nodes \\
Sensitive groups & \texttt{groups} & Node group assignments \\
Group link probability & \texttt{prob\_matrix} & Matrix of link formation probabilities between group pairs \\
Node activity & \texttt{activity\_rate} & Probability a node is active at each time step \\
Popularity influence & \texttt{popularity} & Initial popularity score of each node in link formation \\
Retrain policy & $L$ & Frequency of link prediction updates (periodic or single) \\
Recommendations & \texttt{top\_probs} & Probabilities for the top $k$ recommended candidates \\
\hline
\end{tabular}
\end{table}

\begin{algorithm}[H]
\caption{Pre-network Simulation}
\label{alg:pre_network}
\begin{algorithmic}[1]
\Require $G$, $prob\_matrix$, $top_{probs}$, $popularity$
\State $\mathbf{e}_i \gets \textsc{GetSocialEmbeddings}(nodes)$ \Comment{Latent cluster embeddings, independent of sensitive groups}
\For{$t = 0$ \textbf{to} $T$}
    \State $active\_nodes \gets$ nodes activated according to $activity\_rate$
    \ForAll{$i \in active\_nodes$}
        \State $candidates \gets$ nodes not connected to $i$
        \State $\text{sim}_{ij} \gets \frac{\mathbf{e}_i \cdot \mathbf{e}_j}{\|\mathbf{e}_i\|\|\mathbf{e}_j\|}$ for all $j \in candidates$
        \State $\text{probs} \gets \text{softmax}(\text{sim} - \min(\text{sim}))$ \Comment{Weighted sampling}
        \State $top\_idx \gets$ sample 3 indices without replacement with $p = \text{probs}$
        \For{$rank, idx \in top\_idx$}
            \State $j \gets candidates[idx]$
            \State $p_{\text{sbm}} \gets prob\_matrix[g_i][g_j]$
            \State $acceptance \gets popularity[j]$
            \State $p \gets (p_{\text{sbm}} + top\_probs[rank]) \cdot acceptance$
            \If{$\text{rand()} < p$}
                \State Add edge $(i,j)$ to $G$ with timestamp $t$
            \EndIf
        \EndFor
    \EndFor
\EndFor
\State \Return $G, edges, \mathbf{e}$
\end{algorithmic}
\end{algorithm}

\paragraph{Link prediction models.}
We provide technical details on the link prediction models used in the synthetic experiments.

As non-fairness-aware baselines, we consider node2Vec and a Graph Convolutional Network (GCN). The GCN serves as a backbone encoder to learn node representations from the observed graph, which are then used for link prediction. We additionally consider several fairness-aware methods that operate at different stages of the pipeline (random-walk biasing, graph pre-processing, and debiasing at the representation level). 

\begin{itemize}
    \item \textbf{CrossWalk}~\citep{Khajehnejad_Khajehnejad_Babaei_Gummadi_Weller_Mirzasoleiman_2022} extends random-walk-based node embedding methods (such as node2Vec) by biasing the sampling of random walks so as to encourage transitions across group boundaries. This is achieved by assigning higher weights to edges that lie close to group peripheries or that connect nodes from different sensitive groups, thereby reducing the tendency of walks to remain within the same demographic group.

    \item \textbf{UGE}~\citep{Wang_2022} builds on a GCN backbone and enforces fairness at the representation level by encouraging the learned node embeddings to be invariant with respect to the sensitive attribute. This is achieved through an additional regularization or adversarial objective that discourages the encoder from encoding group-specific information. 

    \item \textbf{FairDrop}~\citep{Spinelli_2022} is a pre-processing strategy that modifies the adjacency matrix during training in order to compensate for homophily with respect to the sensitive attribute. The method relies on a biased edge dropout mechanism: at each training step, edges are removed according to a randomized response scheme that assigns higher removal probabilities to edges connecting nodes with the same sensitive attribute value, thereby attenuating within-group reinforcement in the observed graph.
    \item \textbf{DeBayes}~\citep{buyl2020debayes} is a fairness-aware adaptation of Conditional Network Embedding (CNE)~\citep{kang2019conditional}, a Bayesian framework that incorporates prior knowledge about the network structure (e.g., density, degree distribution, or block structure) through an explicit prior distribution. In DeBayes, sensitive group information is encoded in the prior, such that the learned embeddings are encouraged to be independent of the sensitive attribute while still capturing the underlying structural properties of the graph.
\end{itemize}

For all methods producing node embeddings, link prediction is performed using a cosine similarity-based decoder.
Specifically, for each active node $i$, candidate neighbors are scored by their cosine similarity with $i$ in the embedding space. These scores are then passed through a softmax with temperature $\tau = 0.1$ to obtain a probability distribution over candidates, from which the top-3 candidates are sampled without replacement. This choice ensures that differences in observed behavior across methods primarily stem from the learned representations rather than from the expressiveness of the link prediction head. For GCN-based models, the decoder is applied to the final-layer node embeddings.

Since our objective is not to compare predictive performance but to analyze the qualitative impact of these methods on structural bias and homophily dynamics, we use the default hyperparameters from the original implementations for all methods. For GCN-based models (GCN and UGE), we use the same number of layers, hidden dimensions, and optimizer settings, so that observed differences can be attributed to the fairness mechanisms rather than architectural choices.

\paragraph{Additional results}
Figure \ref{fig:mu} shows the baseline intensity matrix $\hat{\mu}$, estimated on the pre-network phase. The diagonal values confirm the presence of structural homophily before any algorithmic intervention.

\begin{figure}[t]
\centering
\subfloat[\label{fig:mu}]{
  \includegraphics[width=0.23\linewidth]{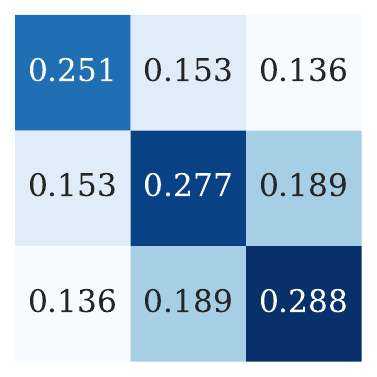}
}\hfill
\subfloat[\label{fig:alpha_noretrain}]{
  \includegraphics[width=0.7\linewidth]{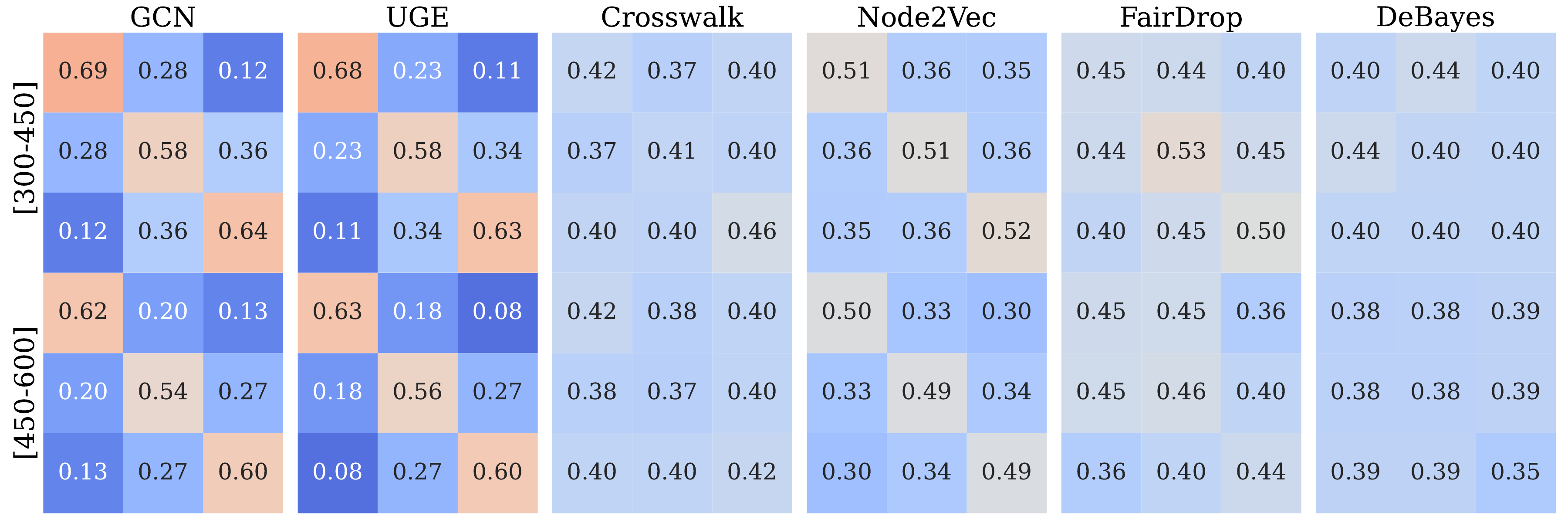}
}
\caption{Estimated Hawkes parameters: (a) $\mu$ estimated; (b) Estimated excitation matrices $\alpha$ over two windows [300--450] and [450--600].}
\label{fig:hawkes_params}
\end{figure}

Figure~\ref{fig:alpha_noretrain} shows the matrices $\alpha$ estimated over two successive windows [300–450] and [450–600], as part of retraining. Compared to the single-window estimates in Figure~\ref{fig:results-synthetic}, the values are significantly higher, which is normal: shorter windows capture a more concentrated signal immediately after each model update, while a single long window dilutes the instantaneous influence of the model over time. GCN and UGE consistently show diagonal coefficients that are highly concentrated on both windows ($0.58$--$0.69$), confirming the persistent algorithmic reinforcement of interactions within the group. 
Crosswalk and DeBayes produce nearly uniform matrices ($0.38$--$0.46$), reflecting low and stable group-level excitation. The models remain largely consistent between the two windows for all models, suggesting that homophilic dynamics stabilize quickly after each intervention. These observations raise a broader question concerning the temporal granularity at which fairness should be evaluated. A coarse measurement window may mask short-term spikes in homophilic excitement that occur immediately after a model update, while a window that is too fine may capture transient noise rather than structural trends. 

\subsection{Social Network (Twitter/X) Application}\label{app:real-world}

\paragraph{Data extraction} We use a retweet network extracted from the dataset released by \cite{Pournaki_Gaisbauer_Olbrich_2025}, covering the year 2021 and the German context. Nodes correspond to user accounts and an edge represents a retweet interaction between two accounts. The original dataset includes topic modeling over tweet content; we restrict our analysis to interactions associated with the german political topic.
The resulting graph contains $70{,}384$ nodes and $581{,}443$ edges.

\paragraph{Labelisation procedure}
To identify the political orientation of users, we adopted a three-step approach. First, keywords characteristic of each political camp were selected to define \textit{anchor} accounts (\textit{seeds}): the most active users on these terms serve as known political reference points. Second, a retweet graph 
is constructed and submitted to the Louvain community detection algorithm, which blindly identifies groups of heavily interconnected users. Finally, the detected communities are labeled by counting the seeds from each political side they contain: the community predominantly populated by CDU seeds is 
tagged \textit{right-wing}, while the one populated by SPD seeds is tagged \textit{left-wing}.

We emphasize that this procedure provides only a reductive proxy for political orientation. Our goal here is not to draw substantive conclusions about political behavior or electoral dynamics, but rather to rely on real-world interaction data to assess whether the proposed model can be learned on non-synthetic networks exhibiting meaningful temporal and structural patterns. The list of seed keywords was established in collaboration with the AI assistant Claude in Table \ref{tab:seeds}.

\begin{table}[h]
\centering
\caption{Keywords used to initialise the political seeds.}
\label{tab:seeds}
\begin{tabular}{ll}
\hline
\textbf{Camp} & \textbf{Keywords} \\
\hline
CDU (right-wing) & \texttt{laschet}, \texttt{maassen}, 
                   \texttt{cduwahlprogramm}, \texttt{union} \\
SPD (left-wing)  & \texttt{scholz}, \texttt{brueckenlockdown}, 
                   \texttt{pimmelgate}, \texttt{spd} \\
\hline
\end{tabular}
\end{table}

\begin{figure}[t]
\centering
\subfloat[\label{fig:mu_germany}]{
\adjustbox{valign=t}{\includegraphics[width=0.23\linewidth]{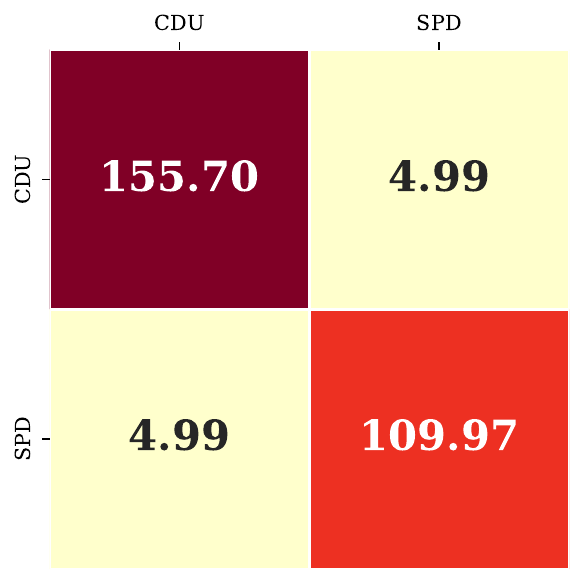}}
}\hfill
\subfloat[\label{fig:alpha_germany}]{
\adjustbox{valign=t}{\includegraphics[width=0.7\linewidth]{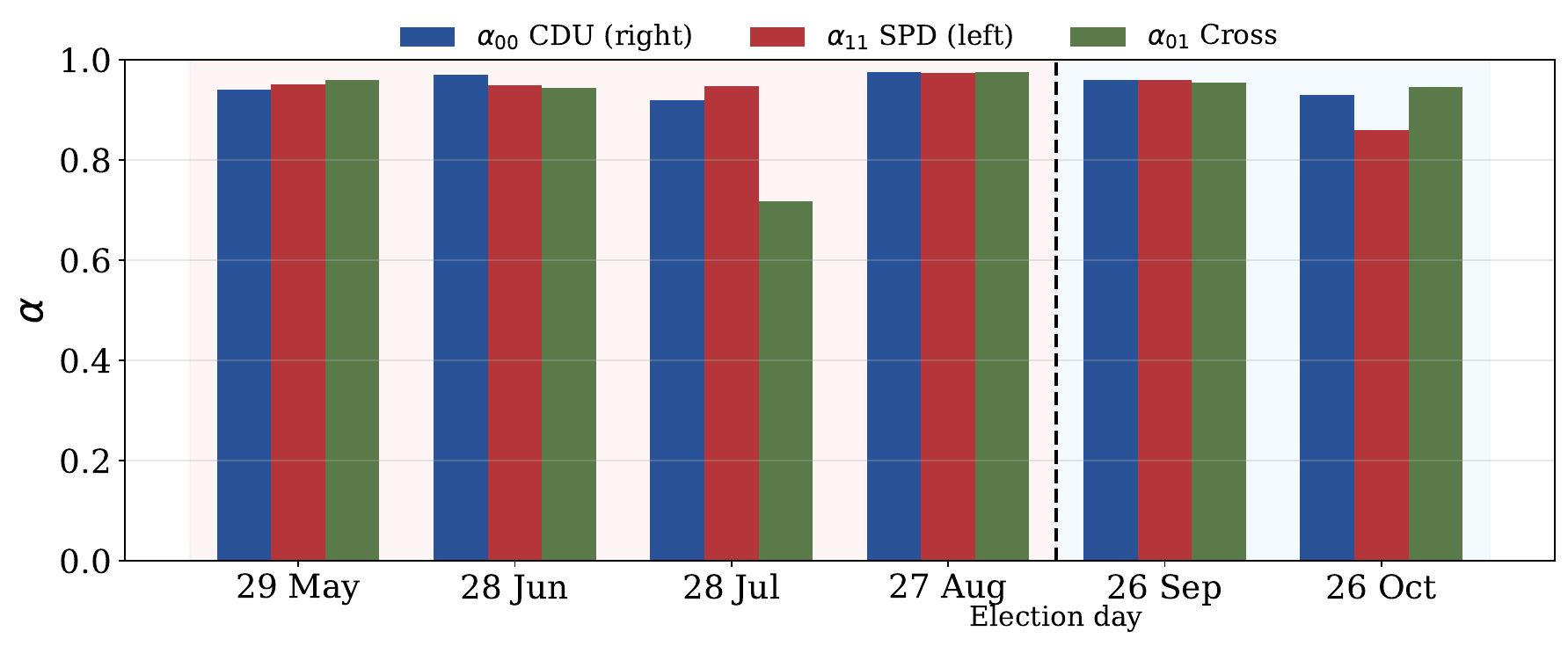}}
}
\caption{Estimated parameters for the German network: (a) $\mu$ estimated; (b) Excitation coefficients computed on the tweet extraction from May to Oct. 2021.}
\label{fig:hawkes_params_germany}
\end{figure}

\paragraph{Additional Results}
Figure~\ref{fig:alpha_germany} shows that the estimated coefficients $\alpha$ are consistently high for all windows, reflecting strong self-excitation within both political camps. Figure~\ref{fig:mu_germany} shows the reference intensity matrix $\mu$ : the cross-group values are about 30 times higher than the within-group term, confirming that the network exhibits structural polarization. These two observations are consistent with the echo chamber effect characteristic of retweet-based platforms, where information circulates mainly within ideologically homogeneous communities. A drop in $\alpha_{01}$ is observed in the July window, coinciding with the Rhineland floods of mid-July. This transient reduction in cross-camp excitation may reflect diverging reactions to the disaster within each political camp, though caution is warranted in drawing causal conclusions from a single window. After the election (26 Sep, 26 Oct), $\alpha$ values remain elevated and stable across all pairs.

\end{document}